\documentclass[10pt,conference]{IEEEtran}
\usepackage{amsmath,amssymb}

\begin{document}

\newtheorem{definition}{Definition}[section]
\newtheorem{thm}{Theorem}[section]
\newtheorem{proposition}[thm]{Proposition}
\newtheorem{lemma}[thm]{Lemma}
\newtheorem{corollary}[thm]{Corollary}
\newtheorem{exam}{Example}[section]

\newtheorem{remark}{Remark}[section]

\newcommand{\La}{\mathbf{L}}
\newcommand{\h}{{\mathbf h}}
\newcommand{\Z}{{\mathbb Z}}
\newcommand{\R}{{\mathbb R}}
\newcommand{\C}{{\mathbb C}}
\newcommand{\D}{{\mathcal D}}
\newcommand{\F}{{\mathbf F}}
\newcommand{\HH}{{\mathbf H}}
\newcommand{\OO}{{\mathcal O}}
\newcommand{\G}{{\mathcal G}}
\newcommand{\A}{{\mathcal A}}
\newcommand{\B}{{\mathcal B}}
\newcommand{\I}{{\mathcal I}}
\newcommand{\E}{{\mathcal E}}
\newcommand{\PP}{{\mathcal P}}
\newcommand{\Q}{{\mathbb Q}}
\newcommand{\separ}{\,\vert\,}
\newcommand{\abs}[1]{\vert #1 \vert}
\newcommand{\mindet}[1]{\hbox{\rm det}_{min}\left( #1\right)}

\title{Diversity-multiplexing Gain Tradeoff: a Tool in Algebra?}
\author{
\authorblockN{Roope Vehkalahti, \textit{Member, IEEE} }
\authorblockA{Department of Mathematics \\
 University of Turku\\
Finland\\
Email: roiive@utu.fi}
\and
\authorblockN{Hsiao-feng (Francis) Lu,  \textit{Member, IEEE}}
\authorblockA{Department of Electronical Engineering \\
National Chiao Tung University\\
Hsinchu, Taiwan\\
Email:francis@cc.nctu.edu.tw }

}
\maketitle

\begin{abstract}
Since the invention of space-time coding numerous algebraic methods have been applied in code design.
 In particular algebraic number theory and central simple algebras have been on the forefront of the research.

In this paper we are turning the table and asking whether information theory can be used as a tool in algebra. We will first derive some corollaries from diversity-multiplexing gain (DMT) bounds by  Zheng  and Tse and later  show how these results can be used to analyze the unit group of orders of certain division algebras. The authors do not claim that the algebraic results are new, but we do find that this interesting relation between algebra and information theory is quite surprising and worth pointing out.

\end{abstract} 

\section{Introduction}
The performance of a lattice code in the Gaussian channel can  be reduced to the considerations of \emph{Hermite constant} and \emph{kissing number}.
In principle capacity results can be used to derive information of achievable Hermite constants and kissing numbers. However, for a given lattice
 in $\C^n$,  with  a given $n$, these results can not be expected to give, for example, tight bounds for Hermite constants.
 This is due to the asymptotic nature of  the classical ergodic capacity results. Performance of  codes with relatively small length is strictly bounded away from capacity.
 
 In the case of fading channels the situation is considerably different. In particular, codes with limited length can achieve the diversity-multiplexing tradeoff bounds. Therefore there is hope that results considering DMT can be transformed into non-trivial mathematical statements considering lattice codes with limited length.
 
 In this paper we are giving some examples how the information theoretic DMT-bounds can be turned into statements of spread of determinants in matrix lattices and how these mass formulas can then be used to analyze unit groups of \emph{orders} of $\Q(i)$-central division algebras.

\section{Basic definitions}
Let us now consider a slow fading channel where we have $n_t$ transmit and $n_r$  receiving antennas and  where the decoding delay is $T$ time units. 
The channel equation can be now written as
$$
Y=\sqrt{\frac{SNR}{n_t}}HX +N
$$
where $H \in M_{n_r \times n_t}(\C)$ is the channel matrix whose entries are independent identically
distributed (i.i.d.) zero-mean complex circular Gaussian
random variables with the variance 1, and $N\in M_{n_r \times T}(\C) $
is the noise matrix whose entries are i.i.d. zero-mean complex circular Gaussian random variables with the variance 1.
Here $X \in M_{n_t\times T}(\C)$ is the transmitted codeword and $SNR$ presents the signal to noise ratio.

In order to shorten the notation we denote $SNR$ with $\rho$.
Let us suppose we have coding scheme where for  each value of $\rho$ we have a code $C(\rho)$ having
$|C(\rho)|$ matrices in $M_{n \times T}(\C)$. The rate $R(\rho)$ is then $\log{(|C(\rho))|}/T$.
Let us suppose  that the scheme fulfills the  constraint  
\begin{equation}\label{energynorm}
\frac{1}{|C(\rho)|}\sum_{X \in C(\rho)} ||X||_F^2 \leq T n_t.
\end{equation}

We then have the following definition from \cite{ZT}.
\begin{definition}
The scheme is said to achieve \emph{spatial multiplexing gain} $r$ and \emph{diversity gain} $d$ if the data rate
$$
\lim_{\rho \to \infty} \frac{R(\rho)}{log(\rho)} = r
$$
and  the average error probability
$$
\lim_{\rho \to \infty} \frac{log(P_e(\rho))}{log(\rho)}=-d.
$$

\end{definition}

\begin{thm}[\cite{ZT}]\label{DMT}

Assume $T\geq m+n-1$. The  optimal tradeoff curve $d^*(r)$  is  achieved by the piecewise-linear function
connecting
$(r, d^*(r)), r=0,\dots,\mathrm{min}(n,m)$, where
$$
d^*(r)=(m-r)(n-r),
$$
and where $r$ is the multiplexing gain.

\end{thm}

Let us now consider a coding scheme based on a $k$-dimensional lattice $L$ inside $M_{n\times T}(\C)$ where for a given positive real number $R$ the finite code is
$$
L(R)=\{a|a \in L,||a||_F\leq R \}.
$$
The following lemma is a  well known result from basic lattice theory. 

\begin{lemma}\label{spherical}
Let $L$ be a  $k$-dimensional lattice in  $M_{n\times T}(\C)$ and
$$L(R)=\{a \,|\, a\in L, \, ||a||_F\leq R \,\},$$
then
$$
|L(R)|= cR^{k}+ f(R),
$$
where $c$ is some real constant and $|f(R)| \in o(R^{(k-1/2)})$. 
\end{lemma}

In particular it follows that we can choose real numbers $K_1$ and $K_2$ so that
\begin{equation}\label{numberofpoints}
K_1R^k\geq |L(R)|\geq K_2 R^k.
\end{equation}

If we then consider a coding scheme where the finite codes are sets
\begin{equation}\label{codingscheme}
C_L( \rho^{rT/k})=\rho^{-rT/k}L(\rho^{rT/k}),
\end{equation}
we will get a  correct number of codewords for each $\rho$ level and the sets $C_L (\rho^{rT/k})$ clearly do fulfill the average energy constraints \eqref{energynorm} expected  in the DMT-analysis (note that here we have not yet added the $\sqrt{\rho}$ needed in the channel equation. Here and in the following we simply forget the term $\frac{1}{n_t}$ in the channel equation as it is irrelevant in DMT calculations.

If we have that $|\det(X)|\geq b$, for all nonzero $X\in L$ and for some constant $b$, we  say that the lattice $L$ has \emph{non-vanishing determinant} (NVD) property \cite{BR}.

\section{Diversity and multiplexing gain trade-off and upper and lower bounds for determinant sums over  matrix lattices}

Let us suppose that we have a $k$-dimensional lattice $L\subseteq M_n(\C)$. 
The finite codes attached to the spherical coding scheme are then 
$$
C_L(\rho^{rn/k})=\rho^{-rn/k}L(\rho^{rn/k}).
$$

\emph{In the following and in the rest of the paper we always suppose that we do not include determinant of the zero matrix to the sum.}

Let us now suppose that we have $n_r$ receiving antennas.
By considering the error probability of transmitting an arbitrary  codeword $X \in C_L(\rho^{rn/k})$ and using the union bound together with PEP based determinant inequality \cite{TSC}, we get the following bound for average error probability for code $C_L(\rho^{rn/k})$
$$
P_e \leq \sum_{X\in L(2\rho^{rn/k})} \frac{\rho^{-nn_r(1-2rn/k)}}{|det(X)|^{2n_r}},
$$ 
where we have used the knowledge of the lattice structure of the  code $L$. In order to take into account that we are considering differences between
codewords we also took the sum over a ball with double radius. We now have
$$
P_e \leq \rho^{-nn_r(1-2nr/k)}\sum_{X\in L(2\rho^{rn/k})} \frac{1}{|det(X)|^{2n_r}},
$$ 
and we can see that the deciding  factor here is  the sum term on right.

To simplify the situation, we will be  considering sums
$$
S_L(R)=\sum_{X\in L(R)} \frac{1}{|det(X)|^{m}}.
$$
Le us now suppose that we have a $k$-dimensional NVD-lattice $L$ in $M_n(\C)$. Let us first give some easy upper and lower bounds for the  asymptotic behavior of  the sums $\sum_{X\in L(R)} \frac{1}{|det(X)|^{m}}$.

Minkowski inequality gives us that 
$$
|det(X)|\leq \left(\frac{||X||_F}{\sqrt{n}}\right)^n.
$$

We then have that 
$$
\sum_{X\in L(R)} \frac{1}{|det(X)|^{m}}  \geq \sum_{||X||_F\leq R, X\in L} \frac{\sqrt{n}^{mn}}{||X||_F^{nm}}.
$$
The right side of this equality is now the beginning of the \emph{Epstein's zeta-function} of the lattice $L$. 
The asymptotic behavior of this function  is well known and we therefore have
$$
\sum_{X\in L(R)} \frac{1}{|det(X)|^{m}}\geq \sum_{||X||_F\leq R, X\in L} \frac{\sqrt{n}^{mn}}{||X||_F^{nm}}  \geq M R^{k-mn},
$$
where $M$ is a constant independent of $R$. 

On the other hand, let us now consider the worst case and suppose that $|det(X)=1|$ for all nonzero $X\in L$ (remember we are working with NVD-lattices). In this case we have 
$$
\sum_{X\in L(R)} \frac{1}{|det(X)|^{m}}=\sum_{X\in L(R)} 1 =|L(R)|\leq N R^k,
$$
where $N$ is a constant independent of $R$ and where the last inequality follows from \eqref{numberofpoints}.

We can now conclude that 
$$
N R^k\geq  \sum_{X\in L(R)} \frac{1}{|det(X)|^{m}}\geq M R^{k-mn},
$$
where $k-mn\geq0$.

Let us now consider the situation where $L$ is a $2n^2$-dimensional lattice in $M_n(\C)$.

In the following proposition we will use the Landau symbol O.
\begin{proposition}\label{mainDMT}
Let us suppose  that we have a  $2n^2$-dimensional NVD-lattice  $L$ in $M_n(\C)$ and that $2|n$. We then have that 
$$
S_L(R)=\sum_{X\in L(R)} \frac{1}{|det(x)|^{2n_r}} \notin O(R^{n^2-\epsilon}),
$$
for any $n_r\geq n$ and positive $\epsilon$.
\end{proposition}
\begin{proof}
Let us use the previously mentioned coding scheme for the lattice $L$.
Just as previously, the  union bound gives us that
$$
P_e \leq \rho^{-nn_r(1-r/n)}\sum_{X\in L(2\rho^{r/2n})} \frac{1}{|det(X)|^{2n_r}}.
$$
The optimal diversity-multiplexing gain given by Zheng and Tse, however, gives us that for integer values of $r$ we have that
$$
P_e \stackrel{.}{\geq}\rho^{-(n-r)(n_r-r)}.
$$
(For dotted notation see \cite{ZT}).
It follows that $S_L(2\rho^{r/2n})$
can not be bounded by 
$$
\rho^{-((n-r)(n_r-r)- nn_r(1-r/n)+\epsilon)  }= \rho^{-(r^2-nr+\epsilon)}
$$
for any positive $\epsilon$, for integer values of $r$. We can now see that the maximum value of $\rho^{-(r^2-nr+\epsilon)}$  is achieved when $r=n/2$.
We then have that 
$$
\sum_{X\in L(2\rho^{(n/2)/2n})} \frac{1}{|det(X)|^{2n_r}}=\sum_{X\in L(2\rho^{1/4})} \frac{1}{|det(X)|^{2n_r}}
$$
can not be bounded by  any $\rho^{n^2/4 -\epsilon}$. When we set  $\rho^{1/4}=R$, we got that   
$S_L(R)$  can not be bounded with  $R^{n^2 -\epsilon}$ for any positive $\epsilon$.
\end{proof}

We can now see that the for $2n^2$-dimensional lattices there exists arbitrarily large values  of $R$ such  that
$S_L(R)\geq R^{n^2 -\epsilon}$, for any $\epsilon$. The most interesting thing here is that no matter how  large $n_r$ we choose this result is valid.
We also see that in some sense the  behavior of the sum is almost the worst possible.

\section{Some results on the unit group of an order in a $\Q(i)$-central division algebra}

\subsection{Problem statement}

Let us suppose that we have a degree $n$ cyclic extension   $E/\Q(i)$ with Galoi's group $ G(E/\Q(i))=<\sigma>$.
 
We can now define a cyclic  algebra  
$$
\D=(E/\Q(i),\sigma,\gamma)=E\oplus uE\oplus u^2E\oplus\cdots\oplus u^{n-1}E,
$$
where   $u\in\mathcal{D}$ is an auxiliary
generating element subject to the relations
$xu=u\sigma(x)$ for all $x\in E$ and $u^n=\gamma\in F^*$.  Let us now suppose that $\D$ is a division algebra.

We can consider $\D$ as a right  vector space over $E$
and every  element $a=x_0+ux_1+\cdots+u^{n-1}x_{n-1}\in\mathcal{D}$
has the following representation as a matrix $\psi(a)=$
\begin{equation}\label{esitys}
\begin{pmatrix}
x_0& \gamma\sigma(x_{n-1})& \gamma\sigma^2(x_{n-2})&\cdots &
\gamma\sigma^{n-1}(x_1)\\
x_1&\sigma(x_0)&\gamma\sigma^2(x_{n-1})& &\gamma\sigma^{n-1}(x_2)\\
x_2& \sigma(x_1)&\sigma^2(x_0)& &\gamma\sigma^{n-1}(x_3)\\
\vdots& & & & \vdots\\
x_{n-1}& \sigma(x_{n-2})&\sigma^2(x_{n-3})&\cdots&\sigma^{n-1}(x_0)\\
\end{pmatrix}.
\end{equation}

\begin{definition}
 A $\Z$-order $\Lambda$ in $\D$ is a subring of $\D$, having the same identity element as
$\D$, and such that $\Lambda$ is a finitely generated
module over $\Z$ and generates $\D$ as a linear space over
$\Q$.
\end{definition}

A simple and easily describable order is  the $\emph{natural order}$ 
$$
\Lambda_{nat}=\OO_E\oplus u\OO_E\oplus u^2\OO_E\oplus\cdots\oplus u^{n-1}\OO_E,
$$
where $\OO_E$ is the ring of algebraic integers in $E$.

This reveals that we can consider  that the ring $\OO_E$ is a subring of the  ring $\Lambda_{nat}$, in particular from the form of the cyclic representation \eqref{esitys}  we can see that $\psi(\OO_E)$ is a sublattice  of $\psi(\Lambda)$  consisting of diagonal elements.

From our perspective the most important properties of these $\Z$-orders are the following
If $\Lambda$ is an $\Z$-order in a division algebra $\D$, then 
$\psi(\Lambda)$ is  $2n^2$-dimensional NVD lattice in $M_n(\C)$, with
$$
|det(X)|\geq 1,
$$
for all the nonzero elements $X$ in $\psi(\Lambda)$.

The unit group $\Lambda^*$ of an order $\Lambda$  consists of elements $x\in \Lambda$ such that there exists an $y\in \Lambda$, such that
$xy=1$. We refer to the unit group of an order $\Lambda$ by $\Lambda^*$.

The unit group  $\OO_E^*$ of the ring of algebraic integers $\OO_E$ is very well known and has simple structure. However, this is not the case for the group $\Lambda^*$. In most cases it is extremely mystical \cite{Kleinert}.

\begin{lemma}
The group $\OO_E^*$ is a normal subgroup  of a unit group $\Lambda^*$  of a any order $\Lambda$ that includes $\OO_E$.

\end{lemma}
\begin{proof}
Clearly  $x(\OO_E)^*=(\OO_E)^*x$, when $x\in E$. For elements $u^k$ we have that
$$
u^k(\OO_E^*)=\sigma^k(\OO_E^*)u^k=(\OO_E^*)u^k,
$$
where the last equality follows from the fact that Galois group operates bijectively on the unit group $\OO_E^*$.
As all the elements of $\D$ are linear combinations of these elements we can see that $\OO_E^*$ is indeed a normal group inside $\Lambda^*$. 
\end{proof}

Due to the normality of the group $\OO_E^*$,  we can for example consider the number of elements $[\Lambda^*:\OO_E^*]$  in the factor group $\Lambda^*/\OO_E^*$.
In this section we are using the simple results  concerning sums of matrix lattices derived from DMT  and we will prove that
$$
[\Lambda^*:\OO_E^*]=\infty.
$$

\begin{remark}
The authors do not suggest that this result is new and it likely follows as a corollary from some more general algebraic result. However, we point out that it is likely not a trivial one. Let us compare  it to another result. This well known and simple result gives us that $[\Lambda^*:\OO_K^*]<\infty$ ($K$ is the center)  if and only if $\D$ is  a totally definite quaternion algebra over  a totally real field. The most simple way to prove this easy result is  to reduce it to  the fact that already $[\OO_E^*:\OO_K^*]=\infty$ (where $E$ is a maximal subfield). The result we are going to prove is considerably stronger and there is no bigger subfield to use as a help. 
\end{remark}

The main idea of our proof is to compare the number of  elements of $\psi(\Lambda^*)\subset M_n(\C)$ and $\psi(\OO_E^*)\subset M_n(\C)$ inside a hypersphere  of radius $R$. We will see that   $\psi(\OO_E^*)$ is not "dense" enough to be a subgroup of finite index in $\psi(\Lambda^*)$.
 
\subsection{Density of units in $\OO_E^*$}
Let us suppose that we have an index $n$ division algebra  $\D =(E/\Q(i),\sigma,\gamma)$.
As previously described in \eqref{esitys} if we now restrict the mapping $\psi$ to the elements of $\OO_E$, we get
an embedding of $\OO_E$ into $M_n(\C)$ by
$$
\psi(x)=\mathrm{diag}(\sigma(x),\dots, \sigma^n(x)),
$$
where $x$ is an element in $\OO_E$.

The ring of algebraic integers $\OO_E$ has a  $\Z$-basis $W=\{w_1,\dots ,w_{2n}\}$ and therefore
$$
\psi(\OO_E)=\psi(w_1)\Z+\cdots +\psi(w_{2n})\Z,
$$
is a $2n$-dimensional lattice of matrices in $M_n(\C)$. For each nonzero element  $a\in \OO_K$, we have that $|det(\psi(a))|\geq 1$.

The unit group $\OO_E^*$ of the ring
$\OO_E$ consists of such elements $u \in \OO_E$, that $|\mathrm{det}(\psi(u))|=1$.

The following lemma is an elementary corollary from well known results. We will skip the proof.
\begin{lemma}\label{units}
Let us suppose that we have a cyclic extension $E/\Q(i)$, where $[E:\Q(i)]=n$. 

We then have that 
$$
|\psi(\OO_E^*)\cap B(R)|\leq M log(R)^{n-1},
$$
where $M$ is a constant independent of $R$.
\end{lemma}

This result proves that the units inside $\OO_E$ are not particularly dense in the lattice $\psi(\OO_E)$. If we consider the
lattice $\psi(\OO_E)$ we have that  $\psi(\OO_E)\cap B(R)$ has roughly $R^{2n}$  elements. The same hypersphere $B(R)$ on the other hand has only 
roughly $log(R)^{n-1}$ units.

\subsection{Density of  the group $\Lambda^*$}

In this section the main main result is Proposition \ref{zeta}, but we need first some results and concepts.
Let us suppose that we have an index $n$ $\Q(i)$-central division algebra $\D$ and that $\Lambda$ is an order in $\D$.
The  (left) \emph{zeta-function} \cite{Sol} of the order $\Lambda$ is
$$
\zeta_{\Lambda}(s)=\sum_{I\in I_{\Lambda}}\frac{1}{[\Lambda:I]^{s}},
$$
where $\Re s> 1$ and $I_{\Lambda}$ is the set of left ideals of $\Lambda$. The fact that we need from this function is that
it is indeed a converging series \cite{BR1}.

The result that will connect this sum  to our matrix lattice considerations is the following
\begin{equation}
|det(\psi(x))|^{2n}= [\Lambda:\Lambda x].
\end{equation}

\begin{lemma}\cite{KW}\label{mismatch}
Let us suppose that $A$ and $B$ are invertible  matrices in  $M_n(\C)$ and that 
$a_1\geq\dots\geq a_n$ are the eigenvalues of $AA^{\dagger}$ and $b_1\leq\dots \leq b_n$ are the eigenvalues of $BB^{\dagger}$.
We then have that
$$
||AB||_F^2\geq \sum_{i=1}^{n} a_i b_i.
$$

\end{lemma}

\begin{lemma}\label{Ux}
Let us suppose that we have a $\Q(i)$-central  division algebra $\D$ with index $n$ and  that $\Lambda$ is an order inside $\D$. If $x\in \Lambda$,  where $||\psi(x)||_F\leq R$, is a non-zero element we have that
$$
|\psi(\Lambda^*x) \cap B(R)|=|\{u \,|\, ||\psi(xu)||_F \leq R, u \in \Lambda^*\}| 
$$
$$
\leq  |\psi(\Lambda^*)\cap B(R^n )|.
$$

\end{lemma}
\begin{proof}
Let us suppose that the eigenvalues of $\psi(x)\psi(x)^{\dagger}$ are $\lambda_1,\dots,\lambda_n$.
The condition $||\psi(x)||_F\leq R$ then gives us that  $\lambda_i\leq R^2 \, \, \forall i$. We also have that 
$ |\lambda_1|\cdots |\lambda_n|\geq 1$. It now follows that
\begin{equation}\label{coordinatesize}
 |\lambda_i|\geq \frac{1}{R^{2(n-1)}}\, \forall i.
 \end{equation}
 Let us now suppose that $u$ is such  a unit  that
$||\psi(ux)||_F=||\psi(u)\psi(x)||_F \leq R$ and let $u_1\geq\cdots \geq u_n$ be the eigenvalues  of $\psi(u)\psi(u)^{\dagger}$. According  to Lemma \ref{mismatch} we then have that
$$
||\psi(u)\psi(x)||_F^2 \geq \sum \lambda_i u_i
$$

Combining equation \eqref{coordinatesize} and $||\psi(u)\psi(x)||_F \leq R$ now gives us  that $||\psi(u)||_F \leq R^n$. 

\end{proof}

\begin{proposition}\label{zeta}
Let us suppose that we have a $\Q(i)$-central index $n$ division algebra $\D$  and that $\Lambda$ is a $\Z$-order in $\D$. We then have
$$
\sum_{||\psi(x)||_F\leq R, x\in \Lambda} \frac{1}{|det(\psi(x))|^{2n n_r}} \leq M |\psi(\Lambda^*)\cap B(R^n)|,
$$
 where $M$ is  independent of $R$.
\end{proposition}
\begin{proof}
The sum
$$ 
\sum_{||\psi(a)||_F\leq R, a\in \Lambda }\frac{1}{|det(\psi(a))|^{2nn_r}}
$$
can be written as
$$
\sum_{x_i\in X}\frac{A_i}{|\mathrm{det}(\psi(x_i))|^{2nn_r}},
$$
where $X$ is some collection of elements $x_i \in \Lambda$, $||\psi(x_i)||_F\leq R$, such that each generate a separate  ideal.  The numbers  $A_i$ present the number of elements inside $B(R)$ each generating the same ideal $x_i\Lambda$. 
We then see that
$$
\sum_{x_i\in X}\frac{1}{|\mathrm{det}(\psi(x_i))|^{2nn_r}}=\sum_{ x_i \in X }\frac{1}{[\Lambda:\Lambda x_i]^{n_r}},
$$
is a part of the zeta-function of  the order $\Lambda$  at point $n_r\geq 2 $. Therefore it is always bounded by some constant $M$ independent of $R$.

From the ideal theory of orders we have that if $\Lambda x_k=\Lambda x_{k'}$, then $x_k$ and $x_k'$ must differ by a unit. Therefore  we can now apply Lemma \ref{Ux} that gives us that  for all $A_i$ we have $A_i \leq |\psi(\Lambda^*)\cap B(R^n )|$. It follows that
$$
\sum_{x_i\in X }\frac{A_i}{[\Lambda:\Lambda x_i]^{n_r}}
$$
$$
\leq \sum_{x_i \in X}\frac{|\psi(\Lambda^*)\cap B(R^n )|}{[\Lambda:\Lambda x_i]}
$$
$$
\leq M|\psi(\Lambda^*)\cap B(R^n )|,
$$
where $M$ is a constant independent of $R$.
\end{proof}

Let us now combine this result with  Proposition \ref{mainDMT}.

\begin{proposition}\label{densityofunits}

Let us suppose that $\Lambda$ is an order in an index $n=2m$ $\Q(i)$-central division algebra $\D$.
We then have that
$$
|\psi(\Lambda^*) \cap B(R)| \notin  O(R^{n-\epsilon}), 
$$ 
for any $\epsilon$.

\end{proposition}

\begin{proof}
We have that $\psi(\Lambda)$ is a $2n^2$-dimensional lattice in $M_n(\C)$.
According to  Proposition \ref{mainDMT} we therefore have that
$$
\sum_{x \in \Lambda,\, ||\psi(x)||_F\leq R} \frac{1}{|det(\psi(x))|^{2n n_r}} \notin O(R^{n^2-\epsilon}) 
$$
for any positive $\epsilon$. On the other hand  Proposition \ref{zeta} gives us that 
$$
\sum_{x \in \Lambda, \,||\psi(x)||_F\leq R} \frac{1}{|det(\psi(x))|^{2n n_r}}\leq M|\psi(\Lambda^*)\cap B(R^n )|,
$$
for some constant independent of $R$. It then follows that 
$$
|\psi(\Lambda^*)\cap B(R )| \notin O(R^{n-\epsilon}).
$$

\end{proof}
This simply means that we can find arbitrarily big $R$ such that  hypersphere $B(R)$ with radius $R$ in $M_n(\C)$ has 
close to $R^n$ elements of $\psi(\Lambda^*)$. On the other hand  $\psi(\Lambda)$ has approximately  $R^{2n^2}$ elements inside the same hypersphere.
While the number of units is small compared to the whole number of points of the lattice, it is still remarkably larger than in the case of number fields where it is in class $(log R)^{n-1}$.

\subsection{A proof that $[\Lambda^*:\OO_E^*]=\infty $}

In this section we are finally giving the proof for the claimed result. We now have the estimates for the number of elements in $\psi(\Lambda^*)$ and 
$\psi(\OO_E^*)$ inside a hypersphere with radius $R$ in $M_n(\C)$. Now we only need some simple results before the  finale.

\begin{lemma}\label{pulla}
Let us suppose that $X$ is a set of matrices in $M_n(\C)$ and that $A$ is an invertible matrix in $M_n(\C)$. If
$f$ is such a function that
$$
|B(R)\cap X |\leq f(R),\, \forall R
$$
then there is such a constant $M$ that
$$
|B(R)\cap AX| \leq \, f(MR),\forall R.
$$
\end{lemma}
\begin{proof}
Let us suppose that $\lambda_1$ is the smallest eigenvalue of $A^{\dagger}A$. According to Lemma \ref{mismatch} we now have that
for all the elements   $Ax \in AX$,  $||Ax||_F^2\geq \lambda_1 ||x||_F^2 $. It follows that
for a matrix  $Ax$, where 
$$
||Ax||_F\leq R,
$$
we must have that $||x||\leq\frac{R}{\sqrt{\lambda_1}}$. We can now see  that $\frac{1}{\sqrt{\lambda_1}}$ is suitable for a constant $M$.

\end{proof}

\begin{proposition}\label{infindeksi}
Let us suppose  $\D=(E/\Q(i),\sigma, \gamma)$ is a cyclic division algebra. Let us suppose that  $\Lambda$ is such an order that it includes
the natural order $\Lambda_{nat}$. We then have that $\OO_E^*$ is a normal subgroup of $\Lambda^*$ and that
$$
[\Lambda^*:\OO_E^*]=\infty.
$$
\end{proposition}

\begin{proof}
Let us suppose that $[\Lambda^*:\OO_E^*]=m$. For certain elements $a_1,\dots, a_m$, we can now write that
$\{a_1\OO_E^*\cup a_2\OO_E^*\cup  \cdots \cup a_8\OO_E^*\}=\Lambda^*$. According to Lemma \ref{units}
 there exists a constant $M$ such that
$$
|\psi(\OO_E^*)\cap B(R)|\leq M (log(R))^{n-1}.
$$
Lemma \ref{pulla} now gives us that there exists constants $M_1,\dots, M_8$ such that
$$
|\psi(a_i\OO_E^*)\cap B(R)|  \leq M log(M_i R)^{n-1}.
$$
As we suppose that $\Lambda^*$ is a union of  $a_i\OO_E^*$, we then have that
$$
|\psi(\Lambda^*)\cap B(R)| \leq \sum_{i=1}^8 M log(M_iR)^{(n-1)}\leq K log (R)^{n-1},
$$
where $K$ is  a constant independent of $R$. However, this is a contradiction against Proposition \ref{densityofunits}.
\end{proof}

\section{Discussion}
The algebraic results we achieved, while interesting, are likely not new. However, the route we used to achieve these results is 
surprising. In our derivation we started with the diversity multiplexing-gain bounds given by Zheng and Tse, which led to
some simple results concerning determinantial sums over matrix lattices and to statement that a unit group of an order is quite "dense". The density result was then applied to derive algebraic results of this group.While some steps where technical the only deep step was taken first. 

The lower bound for asymptotic error probability in the diversity-multiplexing gain tradeoff is coming from the outage probability of
the Rayleigh faded multiple antenna channel. What is needed here is the capacity expression for a MIMO channel and the knowledge of 
the probability density function of singular values of some random matrices.  The final statements of DMT are then gotten by cleverly choosing correct level of approximation that allows one to calculate needed probabilities, but which still gives us nontrivial information of the behavior of the error probabilities of codes in MIMO channel. 

It appears as a lucky accident that we can derive totally algebraic statement from such probabilistic results. 
It is likely that there exists a more direct and probably more effective way  to connect these two areas, but as now  the connection  appear as mystery.

\section*{Acknowledgement}
 The research of  R. Vehkalahti is supported by the Emil Aaltonen Foundation and by the Academy of
Finland (grant 131745). During the making of this paper he was visiting Professor Eva Bayer at \'Ecole polytechnique f\'ed\'erale de Lausanne.

\end{document}